\documentclass[sn-mathphys,fleqn,Numbered]{sn-jnl}

\usepackage{graphicx}%
\usepackage{multirow}%
\usepackage{amsmath,amssymb,amsfonts}%
\usepackage{amsthm}%
\usepackage{mathrsfs}%
\usepackage[title]{appendix}%
\usepackage{xcolor}%
\usepackage{textcomp}%
\usepackage{manyfoot}%
\usepackage{booktabs}%
\usepackage{algorithm}%
\usepackage{algorithmicx}%
\usepackage{algpseudocode}%
\usepackage{listings}%
\usepackage{dsfont}
\usepackage{xspace}
\usepackage[nameinlink]{cleveref}

\newtheorem{theorem}{Theorem}
\newtheorem{lemma}{Lemma}
\newtheorem{corollary}{Corollary}
\newtheorem{claim}{Claim}

\crefname{section}{Section}{Sections}
\crefname{theorem}{Theorem}{Theorem}
\crefname{lemma}{Lemma}{Lemmata}
\crefname{corollary}{Corollary}{Corollaries}
\crefname{claim}{Claim}{Claims}
\crefname{equation}{}{}

\newcommand{\eps}{\varepsilon}
\newcommand{\xmax}{x_{\text{max}}^*}
\newcommand{\xmin}{x_{\text{min}}}

\newcommand{\phalf}{p_{1/2}}
\newcommand{\ptwo}{p_{2/3}}
\renewcommand{\d}{\,\mathrm{d}}
\renewcommand{\epsilon}{\varepsilon}

\def\N{\mathds{N}}
\def\R{\mathds{R}}
\def\CR{\mathrm{cr}}
\newcommand{\kp}{KP\xspace}

\def\opt{\mathrm{opt}}
\def\gain{\mathrm{gain}}
\def\size{\mathrm{size}}

\begin{document}
\title{Online Unbounded Knapsack}

\author*[1]{\fnm{Hans-Joachim} \sur{Böckenhauer}}\email{hjb@inf.ethz.ch}

\author[2]{\fnm{Matthias} \sur{Gehnen}}\email{gehnen@cs.rwth-aachen.de}

\author[1]{\fnm{Juraj} \sur{Hromkovič}}\email{juraj.hromkovic@inf.ethz.ch}

\author[3]{\fnm{Ralf} \sur{Klasing}}\email{ralf.klasing@labri.fr}

\author[1]{\fnm{Dennis} \sur{Komm}}\email{dennis.komm@inf.ethz.ch}

\author[2]{\fnm{Henri} \sur{Lotze}}\email{lotze@cs.rwth-aachen.de}

\author[2]{\fnm{Daniel} \sur{Mock}}\email{mock@cs.rwth-aachen.de}

\author[2]{\fnm{Peter} \sur{Rossmanith}}\email{rossmani@cs-rwth-aachen.de}

\author[1]{\fnm{Moritz} \sur{Stocker}}\email{moritz.stocker@inf.ethz.ch}

\affil[1]{\orgdiv{Department of Computer Science}, \orgname{ETH Zürich}, \city{Zürich}, \country{Switzerland}}

\affil[2]{\orgdiv{Department of Computer Science}, \orgname{RWTH Aachen}, \city{Aachen}, \country{Germany}}

\affil[3]{\orgdiv{CNRS, LaBRI}, \orgname{Université de Bordeaux}, \city{Talence}, \country{France}}

\abstract{
We analyze the competitive ratio and the advice complexity of the online
unbounded knapsack problem. An instance is given as a sequence of $n$ items
with a size and a value each, and an algorithm has to decide whether or not
and how often to pack each item into a knapsack of bounded capacity. The items are given online
and the total size of the packed items must not exceed the
knapsack's capacity, while the objective is to maximize the total value of the packed items.  While each item can only be
packed once in the classical knapsack problem (also called the 0-1 knapsack
problem), the unbounded version allows for items to be packed multiple
times. We show that the simple unbounded knapsack
problem, where the size of each item is equal to its value, allows for a competitive
ratio of 2. We also analyze randomized algorithms and show that, in contrast to the 0-1 knapsack problem, one
uniformly random bit cannot improve an algorithm's performance. More
randomness lowers the competitive ratio to less than $1.736$, but it can never
be below $1.693$.

In the advice complexity setting, we
measure how many bits of information (so-called advice bits) the algorithm has to know to achieve some
desired solution quality. For the simple unbounded knapsack problem,  
one advice bit lowers the competitive ratio to $3/2$. While this cannot be
improved with fewer than $\log_2 n $ advice bits for instances of length $n$, a competitive ratio of $1+\eps$ can be achieved with $O(\eps^{-1} \cdot\log(n\eps^{-1}))$ advice bits for any $\eps>0$. We further show that no
amount of advice bounded by a function $f(n)$ allows an algorithm to be optimal.

We also study the online general unbounded knapsack problem and show that
it does not allow for any bounded competitive ratio for both deterministic and randomized
algorithms, as well as for algorithms using fewer than $\log_2 n$ advice
bits. We also provide a surprisingly simple algorithm that uses 
$O(\eps^{-1} \cdot\log(n\eps^{-1}))$ advice bits to achieve a competitive ratio of $1+\eps$ for any $\eps>0$.
}
\maketitle
\newpage
\section{Introduction}
The knapsack problem is a prominent optimization problem that has been
studied extensively, particularly in the context of approximation
algorithms.  The input consists of items with different sizes and
values and a container, the ``knapsack,\!'' with a given size.  The goal is to select a
subset of the items that fit together into the knapsack, maximizing
their combined value.  Sometimes the name \emph{general knapsack problem} or
\emph{weighted knapsack problem} is used to denote this problem, and sometimes
just \emph{knapsack problem.}  A variant where the values of the items are
proportional to their sizes is sometimes called the \emph{proportional
knapsack problem,} \emph{simple knapsack problem}, or just \emph{knapsack problem}.
That variant is a natural one:  It applies if the items are made of the same
material.  Another popular
name for both variants that emphasizes the fact that every item can be
either rejected or taken once is \emph{0-1 knapsack problem}. In this paper we
will abbreviate the term ``knapsack problem'' with \kp and use the names
\emph{general \kp} and \emph{simple \kp} to avoid any confusion.  

It is well-known that both the general and simple \kp are
weakly NP-hard~\cite{Karp1972,GareyJ1979}.  There is an FPTAS for both problems,
which means that they can be approximated arbitrarily well in polynomial
time~\cite{IbarraK1975,Lawler1979,MagazineO1981,KellererP2004,Jin2019}.

The \kp has many practical applications related to packing resources into constrained spaces. Many of these scenarios take place on an industrial level and it is therefore not unreasonable to assume that items can be packed an essentially unlimited amount of times. We can therefore analyze these situations even more precisely with a variant of the \kp where each item can be packed an arbitrary number of times instead of just once.
This variant is known as the \emph{unbounded
knapsack problem}. Again, the problem is NP-hard and allows for
an FPTAS~\cite{IbarraK1975,KellererPP2004,Lawler1979,JansenK2018}.
In this paper, we investigate the
unbounded \kp as an online problem.  Instead of seeing all
items at once, they arrive one by one and an algorithm has to decide
on the spot whether (and how often) to pack an item into the knapsack or to
discard it. Every decision is final.

Online algorithms are most commonly judged by their \emph{competitive
ratio:}  If we have an input instance $I$ with an optimal solution that
fills a fraction $\opt(I)$ of the knapsack and an algorithm packs only a
fraction $\gain(I)$, then it has a (strict) competitive ratio of $\opt(I)/\gain(I)$
on instance~$I$.  The competitive ratio of the algorithm is the supremum of
all ratios for all possible input instances.  Finally, the competitive
ratio of the problem is the best competitive ratio of all algorithms
that can solve it.  Competitive analysis and the competitive ratio
were introduced by Sleator and Tarjan~\cite{SleatorT1985} and have been
applied to myriads of online problems ever since~\cite{FiatW1998}. For an
overview, we also point to the textbooks by Komm~\cite{Komm2016} and by Borodin
and El-Yaniv~\cite{BorodinE1998}.

Solving both the simple and general 0-1 \kp as online problems is very
hard.  Indeed, there is no algorithm with a bounded competitive ratio.  Such
problems are called \emph{non-competitive}.  It is easy to see that
this is the case~\cite{Marchetti-SpaccamelaV95}:  From now on, we will assume that the capacity of our
knapsack is exactly~$1$.  We look at two different instances for the online simple 0-1 \kp: $I_1=(\epsilon)$ and $I_2=(\epsilon,1)$; so the first
instance consists of only one tiny item $\epsilon>0$, while the second
instance starts with the same item, but follows up with an item that
fills the whole knapsack on its own.  Any algorithm sees first the
item $\epsilon$ and has to decide whether to pack it or discard it.
If the algorithm discards it, the competitive ratio is unbounded for~$I_1$.
If it accepts it, the competitive ratio is very large for $I_2$:  The
optimum is to pack only the second item, which yields a gain of~$1$.
The algorithm, however, packs only $\epsilon$.  The competitive ratio is
therefore $1/\epsilon$, which is arbitrary large.
Note that an average case analysis instead of worst-case analysis paints
a much nicer picture~\cite{Marchetti-SpaccamelaV95,Lueker1998}.

The situation is quite different for the simple unbounded \kp.
There is a simple algorithm that achieves a competitive ratio of~$2$:
Just pack the first item you see as often as possible.  This strategy
will fill more than half of the knapsack.  The optimum can never be
better than~$1$, resulting in a competitive ratio of at most~$2$.
Conversely, the simple unbounded \kp cannot have a competitive ratio
of less than~$2$.  Again, consider two input instances $I_1=(1/2+\epsilon)$
and $I_2=(1/2+\epsilon,1)$.  Any deterministic algorithm will fill the knapsack at most
half-optimal on either $I_1$ or~$I_2$.

This is a first example of a variant of the \kp that is competitive.
There are many other such variants and a curious pattern emerges: Most
of them have also a competitive ratio of~$2$.  For example, there is the
simple \kp with reservations.  In this model, the online algorithm can still pack
or discard an item, but there is also a third possibility:  For an item
of size~$s$, it may pay a reservation cost of $\alpha s$, where $\alpha$ is
a constant -- the \emph{reservation factor}.  Reserved items are not packed, but
also not discarded forever.  They can be packed at a later point in time.
The gain of the algorithm is the combined size of all items minus the
reservation costs.  It is not surprising that the competitive ratio grows
arbitrarily large as we make the reservation cost larger and larger.
Naturally, if the reservation factor exceeds $1$, the problem becomes
non-competitive since the option of reserving becomes useless and we are
back at the classical \kp.  What is surprising, however, is what happens if we
decrease the reservation factor towards~$0$:  In that case, the competitive ratio does
not approach~$1$, but~$2$~\cite{BockenhauerBHLR2021}.

The \kp with removal allows the online algorithm to remove items from the
knapsack.  Once removed, they are discarded and cannot be used in the
future.  Iwama and Taketomi showed that a best possible online algorithm
has the golden ratio $\phi=(1+\sqrt5)/2\approx 1.618$ as its competitive
ratio~\cite{IwamaT2002} for the simple \kp with removal, while the
general \kp with removal remains non-competitive.
The variant where removals are allowed, but
at a cost, has also been investigated~\cite{HanKM14}, as well as a
variant where a limited number of removed items can be packed again~\cite{BoeckenhauerKMRSW2023}.

We have already seen that the simple unbounded \kp has a
competitive ratio of $2$.  To further investigate the problem, we
consider randomized algorithms.  For the simple 0-1 \kp, it is known that
the best possible competitive ratio for randomized online
algorithms is $2$~\cite{BockenhauerKKR2012,BockenhauerKKR2014}.
Surprisingly, a single random bit is sufficient to achieve this ratio
and increasing the number of random bits does not help at all.  For the
simple unbounded \kp, one random bit does not help and the
competitive ratio stays at $2$.  Using two random bits, however, lowers
the competitive ratio to $24/13<1.847$ in \cref{cor:twobits}. We show that no randomized algorithm can
have a competitive ratio of less than $1.693$ in \cref{thm:randlow} and provide a randomized
algorithm with a competitive ratio of less than $1.736$ in \cref{thm:randup}.

Instead of only using classical competitive analysis, Dobrev, Kr\'alovi\v{c}, and
Pardubsk\'a introduced the notion of \emph{advice complexity} of online
algorithms~\cite{DobrevKP2009}, which was later refined by Emek et
al.~\cite{EmekFKR2011} as well as Hromkovi\v{c} et
al.~\cite{HromkovicKK2010} and B\"ockenhauer et
al.~\cite{BockenhauerKKKM17}. In this paper, we use the \emph{tape model} by B\"ockenhauer
et al.  An \emph{online algorithm with advice} is allowed
to read additional information about the instance at hand which is provided by an
all-knowing oracle. The question is how many advice bits have to be read to
achieve a given competitive ratio or even to become optimal. Since its
introduction, the concept of advice complexity has seen a myriad of
applications. Several variants of the \kp  were investigated
\cite{BockenhauerKKR2012,BockenhauerKKR2014,BockenhauerFR2024}. 
For an overview of further advice complexity results, see the textbook by
Komm~\cite{Komm2016} and the survey by Boyar et al.~\cite{BoyarFKLM2017}.
There are also many more recent results. Some of them apply the classical
framework of advice complexity to a wealth of other online problems
including the minimum spanning tree problem \cite{BianchiBBKP2018}, several
variants of online matching
\cite{BockenhauerCU2018,JinM2022,LavasaniP2023}, node and edge deletion
problems \cite{Rossmanith2018,ChenHLR2021,BLotze23}, bin covering
\cite{BoyarFKL2021}, two-way trading \cite{Fung2021}, dominating set
\cite{BoyarEFKL2019,BockenhauerHKU2021}, disjoint path allocation
\cite{BockenhauerKW2022}, or 2D vector packing \cite{NilssonV2021}. Advice
complexity was also used in online-like settings such as exploring a graph
by guiding an autonomous agent
\cite{FraigniaudIP08,DobrevKM12,GorainP19,BockenhauerFU2022} or analyzing
priority algorithms as a model of greedy strategies
\cite{BorodinBLP20,BoyarLP2024,BockenhauerFH2022}.

Two recent strands of research focus on relaxing the condition that the
advice is given by an all-knowing oracle. In the model of untrusted advice
\cite{AngelopoulosDJKR2020,LeeMHLSL2021,AngelopoulosK2023}, one tries to
guarantee that the online algorithm can make good use of the advice as long
as it is of high quality, while the solution quality does not deteriorate
by much in the presence of faulty advice. In a closely related model, a
machine-learning approach is used for generating the advice; the power of this
machine-learned advice has been analyzed for various problems
\cite{AlmanzaCLPR2021,AntoniadisGKK2023,Rohatgi2020,WangLW2020,KasilagRC2022,
IndykMMR2022}. In a very recent paper, Emek et al.~\cite{EmekGPS2023} amend
randomized algorithms by substituting some of the random bits by advice
bits without revealing this to the algorithm. In this paper, we focus only on the classic advice model.

\begin{table}[t]
\caption{New results for the online unbounded \kp.  The table contains the
competitive ratio of the problems.  The $\epsilon>0$ can be
an arbitrarily small constant.  Hence, $1+\epsilon$ indicates that the competitive
ratio can be arbitrarily close to~$1$, while $1$ means that the
algorithm is optimal, and $>1$ means that it \emph{cannot} be optimal.}
\label{tab:results}
\begin{tabular}{c}
\tabskip=0em
\openup0\jot
\def\eps{\epsilon}
\def\unbounded{$\infty$}
\def\frac#1#2{#1/#2}
\leavevmode
\catcode`\!=\active
\def!#1{\rlap{$~^#1$}}
\vbox{\halign{\tabskip=2em\strut#\hfil&\hfil#\hfil&\tabskip=3em\hfil#\hfil&
                     \tabskip=2em\hfil#\hfil&\tabskip=0pt\hfil#\hfil\cr

\noalign{\hrule\smallskip}
&\multispan2\hfil Online 0-1 \kp~\cite{BockenhauerKKR2014}\hfil
                  &\multispan2\hfil Online unbounded \kp\hfil\cr
\noalign{\kern-\medskipamount} 
&\multispan2\hrulefill&\multispan2\hrulefill\cr
\hfil Algorithm&simple&general&simple&general\cr
\noalign{\smallskip\hrule\smallskip}
Deterministic&\unbounded&\unbounded&$2!a$&\unbounded\cr
\noalign{\smallskip\hrule\smallskip}
One random bit&$2$&\unbounded&$2$&\unbounded\cr
Two random bits&$2$&\unbounded&$\leq\frac{24}{13}!b$&\unbounded\cr
Unlimited randomness&$2$&\unbounded&$<1.736!c$&\unbounded$!h$\cr
          &   &          &$>1.693!d$   &          \cr
\noalign{\smallskip\hrule\smallskip}
One advice bit&$2$&\unbounded&$\frac32$&\unbounded\cr
$o(\log n)$ advice&$2$&\unbounded&$\frac32!e$&\unbounded$!i$\cr
$O(\log n)$ advice&$1+\eps$&$1+\eps$&$1+\eps!f$&$1+\eps!j$\cr
$n$ advice bits&$1$&$1$&$>1!g$&$>1$\cr
\noalign{\smallskip\hrule\smallskip\footnotesize
\hbox{\strut
$^a$Thm.~\ref{thm:dettwo},
$^b$Cor.~\ref{cor:twobits},
$^c$Thm.~\ref{thm:randup},
$^d$Thm.~\ref{thm:randlow},
$^e$Thm.~\ref{thm:adviceuptolog},
$^f$Thm.~\ref{thm:simpleptas},
$^g$Thm.~\ref{thm:simpleno1}
}
\hbox{\strut
$^h$Thm.~\ref{thm:generalrandom},
$^i$Thm.~\ref{thm:generaladviceunb},
$^j$Thm.~\ref{thm:genup}
}}
}}
\end{tabular}
\end{table}

For the simple 0-1 \kp, Böckenhauer et al.~\cite{BockenhauerKKR2014} showed that one bit of advice yields a
competitive ratio of $2$.  This ratio cannot be
improved by increasing the amount of advice up to $o(\log n)$, where
$n$ is the number of items.  With $O(\log n)$ advice bits, however, a
competitive ratio arbitrarily close to $1$ can be achieved.  Another
large gap follows:  To achieve optimality, at least $n-1$ advice bits
are needed.  The latter bound has been recently improved by Frei to
$n$ bits~\cite{Frei2021}, which is optimal as $n$ bits can tell the online algorithm
for every item whether to take it or to discard it.

The situation looks similar for the general 0-1 \kp.  We need
$\Theta(\log n)$ advice bits to achieve an arbitrarily good competitive
ratio~\cite{BockenhauerKKR2014}.

For the simple unbounded \kp, we establish the following results:
With one advice bit, the optimal competitive ratio is~$3/2$.  This does
not improve for $o(\log n)$ bits of advice.  Again, with $O(\log n)$
advice bits, a competitive ratio of $1+\epsilon$ for an arbitrarily small
constant $\epsilon>0$ can be achieved.  The online general unbounded \kp stays non-competitive for a sublogarithmic number of advice bits and becomes
$(1+\epsilon)$-competitive with $O(\log n)$ advice bits.
The general unbounded \kp stays non-competitive for randomized algorithms
without advice. \Cref{tab:results} contains an
overview of the new results on the unbounded \kp and compares them to the
known results on the 0-1 \kp.

The paper is organized as follows. In \cref{sec:prelim}, we introduce
all necessary definitions and background information.
\Cref{sec:rand} contains our results on randomized online algorithms
for the simple unbounded \kp. In \cref{sec:advice}, we
consider the advice complexity of the online simple unbounded \kp. \Cref{sec:general} is devoted to the results for the
general case. In \cref{sec:conclusion}, we conclude with a
collection of open questions and reflections on the research area in
general.
Throughout this paper, $\log(\cdot)$ denotes the binary logarithm.

\section{Preliminaries}
\label{sec:prelim}

An online \kp instance of length~$n$ consists of a sequence of items
$x_1,\ldots,x_n$, where the algorithm does not know $n$ beforehand.  Each
item has a size $s_i\in[0,1]$ and a value~$v_i\geq0$;
for the online simple \kp, $s_i=v_i$ for all $1\leq i\leq n$.

An optimal solution is a subset of items that fit together in a knapsack
of size~$1$ and maximize the sum of their values.  We denote the total
value of an optimal solution for an instance~$I$ by
\[
\opt(I)=\max\Bigl\{\,
  \sum_{i\in I'} v_i\Bigm| I'\subseteq I,~\sum_{i\in I'} s_i \leq 1
\,\Bigr\}
\]
for the 0-1 \kp, and by
\[
\opt(I)=\max\Bigl\{\,
  \sum_{i\in I} k_i v_i\Bigm|k_i\in\N_0,~\sum_{i\in I} k_i s_i \leq 1
\,\Bigr\}
\]
for the unbounded \kp.
An online algorithm~$\textsc{Alg}$ maps a sequence of items to a sequence of decisions.  In the 0-1 \kp, the decision is to take the last item or to discard it.  In the unbounded \kp, the algorithm decides how often the last item is packed into the
knapsack.  We define
\begin{align*}
\size_{\textsc{Alg}}(I) &= \text{total size of items packed by~$\textsc{Alg}$ on instance~$I$,}\\
\gain_{\textsc{Alg}}(I) &= \text{total value of items packed by~$\textsc{Alg}$ on instance~$I$,}
\end{align*}
both for the 0-1 \kp and the unbounded \kp.
In the case of a randomized knapsack algorithm, the total value of the packed items is a random
variable.  We then define $\gain_{\textsc{Alg}}(I)$ as the expectation of the total value.

For a deterministic algorithm, we define its (strict) \emph{competitive ratio} on an
instance~$I$ to express the relationship of what it packs to the optimum:
\[
\CR_{\textsc{Alg}}(I)=\frac{\opt(I)}{\gain_{\textsc{Alg}}(I)}
\]
If $\gain_{\textsc{Alg}}(I)$ refers to the expected gain of a randomized algorithm~$\textsc{Alg}$, we speak
of the \emph{competitive ratio in expectation}.

Next, we define the competitive ratio of an algorithm. This is basically the
worst-case competitive ratio over all possible instances.  Finally, the competitive
ratio of an online problem is the competitive ratio of the best
algorithm:
\begin{align*}
\CR_{\textsc{Alg}}=&\sup\{\,\CR_{\textsc{Alg}}(I)\mid \text{$I$ is an instance}\,\}\\
\CR=&\inf\{\,\CR_{\textsc{Alg}} \mid \text{$\textsc{Alg}$ is an algorithm}\;\}
\end{align*}
For a randomized algorithm, the competitive ratio in expectation is defined equivalently.
Computing the competitive ratio is a typical min-max optimization problem,
which can be seen as a game
between an algorithm and an \emph{adversary} who chooses the instance.
In the context of online algorithms with advice, the situation stays
basically the same.  The algorithm, however, has also access to a
binary advice string that is provided by an \emph{oracle}
that knows the whole instance beforehand.  In that way, we can analogously
define the competitive ratio of algorithms with advice.  We say that an
algorithm $\textsc{Alg}$ is using $f(n)$ bits of advice if its reads at most the
first $f(n)$ bits from the advice string provided by the oracle on all
instances consisting of $n$ items.  Note that the length of the advice
string is not allowed to be less than $f(n)$.  Otherwise the oracle
could communicate additional information with the length of the advice
string.

\section{Randomized Algorithms}
\label{sec:rand}
In this chapter we will consider algorithms for the unbounded online \kp that may use randomization, but not advice. As a baseline, we first analyze what can be done without randomness.
A simple observation shows that deterministic algorithms can be
2-competitive for the online simple unbounded \kp, but not better.

\medskip

\begin{theorem}\label{thm:dettwo}
The competitive ratio for the online simple unbounded \kp is~2.
\end{theorem}

\begin{proof}
We start with the upper bound by providing an algorithm with
the claimed competitive ratio which simply packs the first item as often as it fits
into the knapsack. If the first item is $x_1> 1/2$, packing it once achieves a
gain of at least $1/2$. If the first item is $x_1\leq 1/2$, packing it as often
as possible achieves a gain of at least $1-x_1\geq 1/2$. In each case we
achieve a gain of at least $1/2$ and thus a competitive ratio of at most $2$.

Conversely, consider the instances
\[
I_1=\Bigl(\frac{1}{2}+\eps\Bigr),
I_2=\Bigl(\frac{1}{2}+\eps,1\Bigr)
\]
where $0<\eps<1/2$. Any algorithm will either have to pack the item $1/2+\eps$
or not. If it does not, its competitive ratio on $I_1$ is unbounded. If it
does, its competitive ratio on $I_2$ is at most
\[
\frac{1}{1/2+\eps}\,,
\]
which tends to $2$ as $\eps$ tends to $0$.
\end{proof}
In the simple 0-1 \kp, the competitive ratio improves from ``unbounded'' to $2$
when the algorithm is allowed access to one random bit~\cite{BockenhauerKKR2014}.
In the unbounded variant, this is no longer the case.
One random bit is of no use to the algorithm in terms of competitive ratio:

\medskip

\begin{theorem}
An online algorithm for the simple unbounded \kp that uses one uniformly random bit cannot have a competitive ratio of less than $2$.
\end{theorem}

\begin{proof}

Let $\textsc{Alg}$ be any algorithm and once again consider the instances
$I_1$ and $I_2$ from the proof of \cref{thm:dettwo}. In both
instances, $\textsc{Alg}$ will choose the item $1/2+\eps$ with probability
$p$. Since $\textsc{Alg}$ only has access to one random bit, we know
that $p\in\{0,1/2,1\}$.  We have seen in \cref{thm:dettwo}
that the competitive ratio of $\textsc{Alg}$ is at least $2$ on one of
these instances if $p\in\{0,1\}$ and, if $p=1/2$, its competitive
ratio on $I_1$ is
\[
\frac{1/2+\eps}{1/2\cdot(1/2+\eps)}=2\;.
\]
\end{proof}
Additionally, in the 0-1 variant, the competitive ratio of an algorithm did not
improve with additional random bits after the first one
\cite{BockenhauerKKR2014}. This is also no longer true in the unbounded case.
For example,  with two random bits, we can achieve a competitive ratio of
$24/13$ as we will see in \cref{cor:twobits}.

A simple way to prove a lower bound on the competitive ratio of any randomized algorithm is to
provide a set of instances and show that the expected gain has to be
relatively small on \emph{at least one} of those instances \emph{for every}
randomized algorithm.  Let $\epsilon>0$ be very small.  We look at two
instances $I_1=(1/2+\epsilon)$ and $I_2=(1/2+\epsilon,1)$.  The
optimal strategy is to take the first item for~$I_1$ and the second
one for~$I_2$.  Every randomized algorithm takes the first item with some
probability~$p_0$.  Its expected gain on~$I_1$ is therefore
$p_0\cdot(1/2+\epsilon)$ and its expected gain on~$I_2$ is at most
$p_0\cdot(1/2+\epsilon)+(1-p_0)\cdot 1$, because it can pack the second item only if
it ignored the first~one.

Its competitive ratio in expectation is at least
\[
\max\Bigl\{\frac{1}{p_0}, \frac{1}{p_0\cdot(1/2+\epsilon)+1-p_0}\Bigr\}=
\max\Bigl\{\frac{1}{p_0}, \frac{1}{1-p_0(1/2-\epsilon)}\Bigr\}\;.
\]
The competitive ratio of a best possible algorithm is then at least
\[
\min_{p_0\in[0,1]}\,
\max\Bigl\{\frac{1}{p_0}, \frac{1}{1-p_0(1/2-\epsilon)}\Bigr\}
=\frac{3}{2}-\epsilon\;,
\]
and we can conclude that every randomized algorithm has a competitive
ratio of at least $3/2$ since we can make $\epsilon$ arbitrarily
small.

While this lower bound is sound, it is not the best possible.  We can
use the same argument with three instead of two instances in order to
improve it.  Let $I_1=(1/2+\epsilon)$,
$I_2=(1/2+\epsilon,3/4)$, and 
$I_3=(1/2+\epsilon,3/4,1)$.  Using the same argument, the resulting lower bound is 
\begin{equation*}
\min_{\substack{p_0,p_1\geq 0,\\p_0+p_1\leq
1}}\max\left\{\frac{1}{p_0},\frac{3/4}{p_0\cdot (1/2+\eps)+3/4\cdot
p_1},\frac{1}{p_0\cdot (1/2+\eps)+3/4\cdot p_1+1-p_0-p_1}\right\},
\end{equation*}
which converges to $19/12>1.58$ as $\eps\rightarrow
0$, which is better than~$3/2$.  We can still improve this bound by looking at
four instances and so on.

In order to squeeze the last drop out the lemon, we can look at $n+1$
instances like that and let $n$ tend to infinity.  The instances will
be all non-empty prefixes of the sequence
$(1/2+\epsilon,1/2+1/2n,1/2+2/2n,\ldots,1/2+n/2n)$; let $I_k$ denote the
prefix of length $k$, for
$k=0,\ldots,n$.  The above calculation becomes quite complicated as it
involves a rather complicated recurrence relation.  We will therefore
first imagine a continuous variant of the online \kp.
In order to avoid the $\epsilon$, we assume for the moment that an
item of size exactly $1/2$ can be packed only once into the knapsack.
Moreover, we assume that we see an item of size $1/2$ that gradually
grows in size until it reaches some maximum size $s\in[0,1]$ that is
unknown to the algorithm.  An
algorithm sees the item growing and can decide to pack it at any point
in time.  If the algorithm waits too long, however, it will be too late:  As
soon as its size reaches~$s$, the item disappears and the knapsack remains
empty.  As we consider a randomized algorithm, there will be a probability
of $p_0$ that it grabs the item at the very beginning, when its size is
only~$1/2$.  There is also some probability that it will grab the item
before it reaches some size~$x$.  Let $p(t)$ be the density function
of this probability; hence, the probability that the item is taken
before it reaches size~$x$ will be exactly
\[
p_0+\int_{1/2}^x p(t)\,dt\;.
\]
If we look at the instance with maximum size~$s$, the expected gain and
the competitive ratio in expectation are
\[
\frac{p_0}{2} + \int_{1/2}^s tp(t)\,dt\quad\text{and}\quad
s{\Bigm/}\Big(\frac{p_0}{2} + \int_{1/2}^s tp(t)\,dt\Bigr)\;.
\]

To minimize the maximum of the competitive ratio for all~$s$, we choose
$p(t)$ and $p_0$ in such a way that the competitive ratio is the same
for every possible~$s$.  If $s=1/2$, the competitive ratio
becomes~$1/p_0$.  We can determine $p(t)$ by solving the following
equation:
\[
x{\Bigm/}\Big(\frac{p_0}{2} + \int_{1/2}^x tp(t)\,dt\Bigr)=\frac{1}{p_0}\;,
\text{ for all }x\in[{\textstyle\frac12},1]
\]
or, equivalently,
\[
p_0x=\frac{p_0}{2} + \int_{1/2}^x tp(t)\,dt\;.
\]

Taking the derivative of~$x$ on both sides of the equation yields
$p_0=xp(x)$, and we have~$p(t)=p_0/t$ for $t\in[\frac12,1]$.  It
remains to determine~$p_0$.  To this end, we can use the additional
equation
\begin{equation}\label{eq:estimate_p0}
1=p_0+\int_{1/2}^1 p(t)\,dt=
p_0+\int_{1/2}^1 \frac{p_0}{t}\,dt=p_0(1+\ln 2) \;,
\end{equation}
which results from the fact that the total probability is $1$.

In this game of grabbing continuously growing items, the competitive ratio of our specific algorithm is $1+\ln 2$.
However, this was
just a thought experiment, which gives us a clue on how to choose
the probabilities $p_0,\ldots,p_n$ for a best possible algorithm of the
instances $I_0,\ldots,I_n$ defined earlier.  They should be approximately
$p_k\approx p(1/2+k/2n)/2n= p_0/(n+k)= 1/((1+\ln 2)(n+k))$,
where the first equality follows from $p(t)=p_0/t$ and the second equality
follows from~\cref{eq:estimate_p0}.
Moreover, we have to show that every other algorithm cannot be better
than this one.  To this end, we have to show that this choice of the values
$p_k$ is optimal and every other choice leads to a worse competitive
ratio for at least one of the instances.

\medskip

\begin{theorem}\label{thm:randlow}
The competitive ratio in expectation of every randomized algorithm
that solves the online simple unbounded \kp is at least $1+\ln 2>1.693$.
\end{theorem}

\begin{proof}
We consider the instances
\begin{align*}
I_0 ={}& \Bigl(\frac12+\epsilon\Bigr),\\
I_1 ={}& \Bigl(\frac12+\epsilon,~\frac12+\frac{1}{2n}\Bigr),\\
I_2 ={}& \Bigl(\frac12+\epsilon,~\frac12+\frac{1}{2n},~
               \frac12+\frac{2}{2n}\Bigr),\\
I_3 ={}& \Bigl(\frac12+\epsilon,~\frac12+\frac{1}{2n},~
               \frac12+\frac{2}{2n},~\frac12+\frac{3}{2n}\Bigr),~\\
\vdots&\\
I_n ={}& \Bigl(\frac12+\epsilon,~\frac12+\frac{1}{2n},~
               \frac12+\frac{2}{2n},~\frac12+\frac{3}{2n},~
               \ldots\ldots, \frac12+\frac{n}{2n}\Bigr),
\end{align*}
where $n$ is large and $\epsilon$ is very small, say, $\epsilon=e^{-n}$.
We design an algorithm that will work as follows:  If it sees an item
whose size is in the interval $[1/2+k/2n,1/2+(k+1)/2n)$, where
$0\leq k\leq n$, it will accept it with a probability of $p_k$, where
\[
p_0=\frac{1}{1+H_{2n}-H_n} \text{ and }
p_k=\frac{1}{(1+H_{2n}-H_n)(n+k)}\;, \text{ for $k>0$}\;.
\]
We use the difference of Harmonic numbers $H_{2n}-H_n=\ln 2+O(1/n)$ instead of $\ln 2$ to make the sum
of all probabilities exactly one:
\[
\sum_{k=1}^n p_k = \frac{1}{1+H_{2n}-H_n}\sum_{k=1}^n \frac{1}{n+k} = 
\frac{1}{1+H_{2n}-H_n} (H_{2n}-H_n)=1-p_0\;.
\]

What is the expected gain of this algorithm on~$I_k$?  It will accept
the $i$th item of size $1/2+i/2n$ (possibly${}+\epsilon$) with a
probability of $1/\bigl((1+H_{2n}-H_n)(n+i)\bigr)$.
So the expected gain turns out to be exactly
\[
p_0\cdot\Bigl(\frac12+\epsilon\Bigr) + \frac{1}{1+H_{2n}-H_n}
   \sum_{i=1}^k\Bigl(\frac{n}{2n}+\frac{i}{2n}\Bigr)\frac{1}{n+i}
=\frac{1/2+k/(2n)+\epsilon}{1+H_{2n}-H_n}\;.
\]
The optimal gain for $I_k$ is $1/2+k/2n$ (or $1/2+\epsilon$ if
$k=0$).  Hence, the competitive ratio in expectation on \emph{every} $I_k$
is $1+H_{2n}-H_n+O(\epsilon)$, which goes to $1+\ln(2)$ as $n\rightarrow\infty$ and $\eps\rightarrow 0$.

Can there be a different algorithm with better competitive ratio?  If
yes, it has to beat our algorithm on \emph{every} instance $I_k$, since our algorithm has the same competitive ratio on those instances.
We will show that, for an algorithm whose competitive ratios differ on the
$I_k$'s, there is another, not worse algorithm whose competitive ratios
are the same.  We prove this by stepwise transforming such a non-uniform
algorithm.  So let us fix some algorithm for the online simple unbounded
\kp whose competitive ratio in expectation on $I_k$ is~$c_k'$
and whose probability of taking the $i$th item in $I_n$ is~$p_i'$.

The following result tells us what happens if we change an algorithm by
moving some probability from $p_i'$ to~$p_{i+1}'$.

\begin{claim}
\label{claim:changeprobs}
  Let $0\le i<n$. Let $p_i''=p_i-\delta$, $p_{i+1}''=p_i'+\delta$, and $p_j''=p_j'$ for
  $j\notin\{i,i+1\}$ for some $\delta\in\R$ with sufficiently small
  absolute value. Let $c_i''$ be the competitive ratios of the algorithm
  $A''$ that uses the probabilities~$p_i''$.

Then $c_j''=c_j'$, for $j=0,\ldots,i-1$.

If $\delta>0$, then $c_i''<c_i'$ and $c_j''>c_j'$, for $j>i$.

If $\delta<0$, then $c_i''>c_i'$ and $c_j''<c_j'$, for $j>i$.
\end{claim}
\begin{proof}
Obviously, $c_j''=c_j'$ for $j<i$.  Let us look at $c_i''$ and $c_i'$
for $\delta>0$:
\[
\let\~=\displaystyle
c_i''=\frac{1+i/n}{\!\!\~p_0+\sum_{j=1}^i\Bigl(1+\frac jn\Bigr)p_i''\!\!}
    =\frac{1+i/n}{\~p_0+\sum_{j=1}^i\Bigl(1+\frac jn\Bigr)p_i'+
    \Bigl(1+\frac in\Bigr)\delta}<
\frac{1+i/n}{\!\!\~p_0+\sum_{j=1}^i\Bigl(1+\frac jn\Bigr)p_i'\!\!}=c_i'
\]
A similar calculation can be done for $c_k''$ and $c_k'$ for $k>i$:
\begin{multline*}
\let\~=\displaystyle
c_k''=\frac{1+k/n}{\~p_0+\sum_{j=1}^k\Bigl(1+\frac jn\Bigr)p_i''}
    =\frac{1+k/n}{\~p_0+\sum_{j=1}^k\Bigl(1+\frac jn\Bigr)p_i'+
    \Bigl(1+\frac in\Bigr)\delta-\Bigl(1+\frac{i+1}n\Bigr)\delta}\\
    \let\~=\displaystyle
    =\frac{1+k/n}{\~p_0+\sum_{j=1}^k\Bigl(1+\frac jn\Bigr)p_i'-
    \Bigl(1+\frac 1n\Bigr)\delta}>
\frac{1+k/n}{\~p_0+\sum_{j=1}^k\Bigl(1+\frac jn\Bigr)p_i'}=c_k'
\end{multline*}

The calculation for $\delta<0$ is completely analogous.
\end{proof}

Using \cref{claim:changeprobs}, it is relatively easy to modify an algorithm
whose competitive ratios $c_i'$ differ from another.  We look at
those $c_i'$ that are maximum.  Let $c_i'$ be the last one that has
maximal value.  If $i<n$, we can use a small $\delta>0$ and increase
$p_i'$ by $\delta$ and decrease $p_{i+1}'$ by~$\delta$.  If $\delta$ is
small enough, then $c_i''<c_i'$ while still $c_i''>c_j''$ for $j>i$.
If $c_i'$ was the only maximum ratio, we improved the algorithm.  If
not, it is no longer the last one and we can apply the same procedure
again.  Eventually, the algorithm will improve.

There is, however, the other possibility that the last ratio $c_n'$
is a maximum one.  Then $c_{n+1}'$ does not exist and we cannot do the
above transformation.  In this case, however, we choose the last
ratio $c_i'$ where $c_i'<c_n'$ (so $c_{i+1}'=c_{i+2}'=\cdots=c_n'$).
We decrease $p_i'$ by a small $\delta$ and increase $p_{i+1}'$ by the
same~$\delta$.  If $\delta$ is small enough then still $c_i''<c_j''$
and $c_j''<c_j'=c_n'$ for all $j>i$.  Again, either the new algorithm
has a better competitive ratio or the maximal $c_i''$ is now not at
the right end and the first case applies.

In this way, after finitely many steps, we either get to an algorithm
that is better or an algorithm where the competitive ratios agree on all instances $I_0,\dots,I_n$.
\end{proof}
Note that what an algorithm must do on this family of instances is choose the largest of a series of items. This problem is known as the Online Search Problem and has been studied in much detail \cite{BorodinE1998}. In the classic variant of the problem, an algorithm would be allowed (or forced) to keep the final item presented, while in our model, a gain of $0$ is possible. Thus, any lower bound for the Online Search Problem on items in the interval $]1/2,1]$ also holds for the online unbounded \kp. However, even deterministic search algorithms on that interval can achieve a competitive ratio of $\sqrt{2}< 1.415$ (see for example Chapter 14.1.2 of \cite{BorodinE1998}). The lower bound we provide is considerably stronger.
We now complement this bound by a rather close upper bound.

\medskip

\begin{theorem}
\label{thm:randup}
There is a randomized algorithm that solves the online simple unbounded \kp with
a competitive ratio of less than $1.7353$. 
\end{theorem}
\begin{proof}
The algorithm computes a random variable $X$ with the following distribution:
\begin{itemize}
	\item $\Pr[X=1/2]=\phalf$
	\item $f_X(x)=\phalf/x$ \quad if $1/2<x<2/3$
	\item $\Pr[X=2/3]=\ptwo$
	\item $f_X(x)=\phalf\cdot(1+\ln(2-x))/2x$ \quad
              if $2/3<x\leq 1$
\end{itemize}
where $\phalf$ and $\ptwo$ solve the system of equations
\begin{align}
\label{eq:randomprob1}
\phalf&=\frac{2\ptwo}{3}+\frac{2\phalf}{3}+\frac{2\phalf}{3}\cdot\ln(4/3)\\
\label{eq:randomprob2}
1&=\phalf+\phalf\cdot\ln(4/3)+\ptwo+\phalf\cdot\int_{2/3}^1\frac{1+\ln(2-x)}{2x}\,dx
\end{align}

The fact that this is indeed a
probabilistic distribution is guaranteed by \Cref{eq:randomprob2}, while \cref{eq:randomprob1} will be
used to prove bounds on the competitive ratio.
After choosing $X$, the
algorithm packs as often as possible the first item $x$ it encounters with $x^*\geq X$, where $x^*=\lfloor 1/x\rfloor \cdot x$ is defined as the gain that can be achieved by the item $x$ alone. After that, it packs any item
that fits greedily.

We now look at a fixed instance $I=(x_1,\dots,x_n)$. We define $\xmin=\min_i \{x_i\}$ and $\xmax=\max_i \{x_i^*\}$.
If $\xmax<2/3$, there cannot be any items smaller than $1/2$ (packing
these items as often as possible would lead to $\xmax>2/3$). The optimal solution is thus $\xmax$. The algorithm has a gain of at least $X$ if $X\leq \xmax$, so an expected gain of at least 
\[
\frac{\phalf}{2}+\int_{1/2}^{\xmax} x\cdot f_X(x)\d x=
\frac{\phalf}{2}+\phalf\cdot\Bigl(\xmax-\frac{1}{2}\Bigr)=\phalf\cdot \xmax.
\]
Its competitive ratio in expectation is thus at most $1/\phalf$.

If $\xmax\geq2/3$ and $\xmin\geq1/2$, the optimal solution is still
$\xmax$. The algorithm has a gain of at least $X$ if $X\leq \xmax$, and
thus a competitive ratio of at most
\begin{multline*}
\xmax\!\biggm/\!\!\biggl(\frac{\phalf}{2}+\int_{1/2}^{2/3} x\cdot
f_X(x)\,dx+\frac{2p_{2/3}}{3}
+\int_{2/3}^{\xmax}x\cdot f_X(x)\,dx\biggr)\\
=\xmax\!\biggm/\!\!\biggl(\frac{2\phalf}{3}+\frac{2p_{2/3}}{3}
+\int_{2/3}^{\xmax}x\cdot f_X(x)\,dx\biggr).
\end{multline*}
It can be shown using standard techniques from calculus that this ratio increases in $\xmax$. This is done in full detail in \cref{lemma:increase} in the appendix. Hence this bound is maximized for $\xmax=1$ and is therefore at most
\begin{multline*}
\biggl(\frac{2\phalf}{3}+\frac{2\ptwo}{3}+\frac{\phalf}{2}
\cdot\int_{2/3}^1 (1+\ln(2-x))\,dx\biggr)^{-1}\\
=\biggl(\frac{2\phalf}{3}+\frac{2\ptwo}{3}
+\frac{2\phalf}{3} \cdot \ln(4/3)\biggr)^{-1}
=\frac{1}{\phalf}
\end{multline*}
by \cref{eq:randomprob1}.

If $1/3<\xmin<1/2$ and $X\leq 1-\xmin$, the algorithm will always achieve a gain of at least $1-\xmin$. When it encounters $\xmin$ it either has not packed any items yet, in which case it packs $\xmin^*=2\xmin\geq 1-\xmin\geq X$, or it has already packed other items. In this case, it packs $\xmin$ greedily as often as possible and still achieves a gain of at least $1-\xmin$. If on the other hand $1-\xmin<X\leq 2\xmin$, the algorithm will always achieve a gain of at least $X$. When it encounters $\xmin$ it has already packed items of size at least $X$ or it will be able to pack $\xmin^*=2\xmin\geq X$. Its expected gain is therefore at least
\begin{multline*}
(1-\xmin)\cdot\Pr[X\leq 1-\xmin]+E[X\mid 1-\xmin< X\leq 2\xmin]=\\
(1-\xmin)\biggl(\int_{1/2}^{1-\xmin}\!\!\!\!\!\!f_X(x)\,dx+p_{1/2}\biggr)
+\int_{1-\xmin}^{2/3}\!\!\!\!\!\!\!\!
xf_X(x)\,dx+\frac{2p_{2/3}}{3}+\int_{2/3}^{2\xmin}\!\!\! xf_X(x)\,dx.
\end{multline*}
With $y=1-\xmin$, we can simplify this to
\begin{multline*}
y\cdot\Bigl(\phalf\cdot \ln(2y)+\phalf\Bigr)
+\phalf\Bigl(\frac{2}{3}-y\Bigr)
+\frac{2p_{2/3}}{3}+\int_{2/3}^{2-2y} xf_X(x)\,dx=\\
=y\cdot\Bigl(\phalf\cdot \ln(2y)+\phalf\Bigr)
+\phalf\Bigl(\frac{2}{3}-y\Bigr)+\frac{2p_{2/3}}{3}
+\frac{\phalf}{2}\int_{2/3}^{2-2y} \Bigl(1+\ln(2-x)\Bigr)\,dx\\
=y\cdot\Bigl(\phalf\cdot \ln(2y)+\phalf\Bigr)
+\phalf(\frac{2}{3}-y)+\frac{2p_{2/3}}{3}
+\frac{\phalf}{2}\biggl(-2y\ln(2y)
+\frac{4}{3}\ln\Bigl(\frac{4}{3}\Bigr)\biggr)\\
=\frac{2\ptwo}{3}+\frac{2\phalf}{3}
+\frac{2\phalf}{3}\cdot\ln\Bigl(\frac43\Bigr)=\phalf
\end{multline*}
by \cref{eq:randomprob1}. Its competitive ratio is therefore
at most $1/\phalf$.

If $\xmin\leq1/3$, we argue similarly to the previous case. If $X\leq
1-\xmin$, the algorithm achieves a gain of at least $1-\xmin$ and if
$1-\xmin\leq X\leq \lfloor 1/\xmin\rfloor\xmin$, it has a gain of at
least $X$. Again with $y=1-\xmin$ and since $y\geq 2/3$, this leads to an expected gain of at least 
\begin{multline*}
y\biggl(\phalf+\int_{1/2}^{2/3}f_X(x)\,dx+\ptwo+\int_{2/3}^{y}f_X(x)\,dx\biggr)
+\int_{y}^{\lfloor 1/\xmin\rfloor\xmin} xf_X(x)\,dx\\
\geq y\biggl(\phalf+\int_{1/2}^{2/3}f_X(x)\,dx+\ptwo\biggr)
\geq\frac{2}{3}\biggl(\phalf+\phalf\ln\Bigl(\frac{4}{3}\Bigr)+\ptwo\biggr)
=\phalf
\end{multline*}
by \cref{eq:randomprob1}.

So in each case the competitive ratio of the algorithm is at most
$1/\phalf$.  Solving the system of \cref{eq:randomprob1,eq:randomprob2} for $\phalf$ and $\ptwo$ leads to 
\[
\tabskip=0pt\def\!{\displaystyle}
\vbox{\halign{\hfil$\!#$&\hfil$\!#$&$\!#$\hfil\cr
\phalf=&2\;&\biggl(3+\int_{2/3}^1 \frac{1+\ln(2-x)}{x}\,dx\biggr)^{-1}\cr
\ptwo=&\Bigl(1-2\ln\Bigl(\frac{4}{3}\Bigr)\Bigr)
       &\biggl(3+\int_{2/3}^1 \frac{1+\ln(2-x)}{x}\,dx\biggr)^{-1}\cr
}}
\]
so the algorithm has a competitive ratio of at most 
\begin{equation*}
\frac{1}{\phalf}=\frac12\biggl(3+\int_{2/3}^1
\frac{1+\ln(2-x)}{x}\,dx\biggr)< 1.7353.
\end{equation*}
\end{proof}

\section{Advice Complexity}
\label{sec:advice}

Recall that in the online simple 0-1 \kp the competitive ratio turned out to be unbounded without advice and
$2$~for one and up to $o(\log n)$ many advice bits.  With $O(\log n)$
advice bits, the competitive ratio is arbitrarily close
to~$1$~\cite{BockenhauerKKR2014}.
For the online simple unbounded \kp, the situation is almost analogous.  With one
advice bit, the competitive ratio is~$3/2$ and again stays at $3/2$
for $o(\log n)$ advice bits.  Then, with $O(\log n)$ many advice bits
we again come arbitrarily close to~$1$.

\medskip

\begin{theorem}
\label{thm:adviceuptolog}
There is an algorithm for the online simple unbounded \kp that reads only a single bit of advice and achieves a competitive ratio of $3/2$. For any fixed number $k>0$ of advice bits, this ratio cannot be improved. To achieve a competitive ratio of better than $3/2$ on all instances of length $n$, an algorithm must read more than $\log(n-1)$ advice bits on at least one such instance.

\end{theorem}

\begin{proof}
Let us first see why we can achieve a competitive ratio of $3/2$ with a
single advice bit.  An algorithm can use the advice bit to choose between
two strategies.  Here, those two strategies are: (1)~pack greedily and
(2)~wait for a item that is either smaller than or equal to ~$1/2$ or is
greater than or equal to $2/3$ and pack it as often as possible.

We have to show that, for every possible instance, one of the two
strategies manages to pack at least a $2/3$-fraction of the optimal
solution.  So let us assume that there indeed exists at least one item
for which strategy~(2) is waiting.  If its size is at least $2/3$, the
algorithm has a gain of at least $2/3$.  If its size $s$ is smaller than or equal to $1/2$,
the same is true:  After packing the item as often as possible,
the remaining space in the knapsack will be smaller than $s$, but will
contain at least two items of size~$s$.  So if $s\geq 1/3$, the two items
sufficiently fill the knapsack, and if $s\leq 1/3$, the remaining space is less than~$1/3$. In either case, the algorithm achieves a gain of at least $2/3$ and thus a competitive ratio of at most $3/2$.

So strategy~(2) succeeds whenever the instance contains an item
of size $s$ with $s\leq 1/2$ or $s\geq2/3$.  Let us assume now that the
input contains no such item.  Then strategy~(1) succeeds:  All items have
sizes strictly between $1/2$ and~$2/3$.  This means in particular that
only one item fits into the knapsack and the optimal solution fills the
knapsack at most to a level of~$2/3$.  The greedy algorithm obviously
fills the knapsack at least half.  The competitive ratio is then
at most $(2/3)/(1/2)=4/3$.

In order to prove a lower bound, let $0<\epsilon<1/2$ and we consider
$n-1$ different instances $I_2,\ldots,I_n$ where
\[
I_k=\Bigl(\frac13+\epsilon,\frac13+\epsilon^2,\frac13+\epsilon^3\ldots,
     \frac13+\epsilon^{k-2},
     \frac13+\epsilon^{k-1},\frac23-\epsilon^{k-1},\dots,\frac23-\eps^{k-1}\Bigr)\,,
\]
ending in $n-k+1$ items of size $2/3-\eps^{k-1}$, such that each instance consists of $n$ items.
It is easy to see that the optimal gain for every instance is~$1$.
To be optimal on $I_k$,
the algorithm has to reject the first $k-2$ items, and pack the $(k-1)$-th item, as well as one of the larger items of size $2/3-\eps^{k-1}$.
Intuitively, an algorithm that receives fewer than $\log(n-1)$ advice bits cannot choose
between $n-1$ or more different strategies. More formally, assume that the algorithm reads fewer than $\log(n-1)$ advice bits on all $n-1$ instances. Since there are fewer than $n-1$ binary strings of length less than $\log(n-1)$, at least two instances $I_i$ and $I_j$ for $2\leq i<j\leq n$ must receive the same advice string. Any decision of the algorithm, including the amount of advice bits to read, can only depend on the prefix of the instance and the advice bits read at that point. This means that, after the common prefix $(1/3+\eps,\dots,1/3+\eps^{i-1})$, the algorithm must make the same decision in $I_i$ and $I_j$. However, to be optimal on $I_i$, it must take the item of size $1/3+\eps^{i-1}$, while to be optimal on $I_j$, it must not take it.
Therefore, it will behave suboptimally on at least one of the two instances.  For that
instance, its competitive ratio will be at least $1/(2/3+2\epsilon)$,
because the algorithm can pack only two small items or one large one.
As $\epsilon$ can be chosen arbitrarily small, the competitive ratio of any
algorithm cannot be better than~$3/2$.

\end{proof}
\begin{corollary}
\label{cor:guessadvice}
For any $0<p<1$, there is a randomized algorithm for the online simple unbounded \kp that makes a single random decision with probabilities $p$ and $1-p$ and has a competitive ratio of at most
\[
\max\biggl\{\Bigl(p\cdot\frac{1}{2}+(1-p)\cdot
\frac{2}{3}\Bigr)^{-1},~\frac{2/3}{p/2}\biggr\}.
\]
\end{corollary}
\begin{proof}
The algorithm simulates the algorithm with one advice bit in the proof of
\cref{thm:adviceuptolog} and guesses the advice bit. With a probability of
$p$, it packs items greedily and, with a probability of $1-p$, it waits for
an item $x_i$ with $x_i\leq 1/2$ or $x_i\geq 2/3$ and then packs $x_i$ as
often as possible. If there is such an item, its competitive ratio in
expectation is at most
\[
1{\Bigm/}\Bigl(p\cdot\frac{1}{2}+(1-p)\cdot\frac{2}{3}\Bigr).
\]
If there is no such item, the optimal solution is at most $2/3$ and the
algorithm's competitive ratio in expectation is at most
$(2/3){\bigm/}(p/2)$.
\end{proof}

\begin{corollary}
\label{cor:twobits}
There is a randomized algorithm for the online simple unbounded \kp that uses exactly two uniformly random bits and has a competitive ratio of at most $24/13$.
\end{corollary}
\begin{proof}
The algorithm uses two uniformly random bits to generate a random decision with probability $p=3/4$ and uses the algorithm from \cref{cor:guessadvice}.
\end{proof}
\begin{corollary}
\label{cor:onechoice}
There is a randomized algorithm for the online simple unbounded \kp that makes a single non-uniformly random decision and has a competitive ratio of at most $11/6$.
\end{corollary}
\begin{proof}
This is the algorithm from \cref{cor:guessadvice} for $p=8/11$.
\end{proof}

\Cref{thm:adviceuptolog} is tight.  If we provide $O(\log n)$ advice, the
competitive ratio can be made arbitrarily close to optimal.

\medskip

\begin{theorem}
\label{thm:simpleptas}
Let $\epsilon>0$.  Then there is an algorithm using
$O((1/\epsilon)\log(n/\epsilon))$ advice
bits that solves the online simple unbounded \kp with a
competitive ratio of at most \mbox{$1+\epsilon$}.
\end{theorem}

\medskip

This is a special case of the algorithm we will see in \cref{thm:genup} for the online general unbounded \kp. 

An interesting phenomenon can be observed when we ask about the
necessary number of advice bits to solve the problem exactly, i.e., to
reach a competitive ratio that is exactly~$1$.  Of course, $n$ advice
bits can achieve this in the 0-1 model:  Just use one advice
bit for each item, which tells us whether to take or discard it.
It also turned out that one really needs $n$ bits \cite{Frei2021} to be optimal.

Surprisingly, the same argument does not work for the unbounded \kp.
One bit can tell us whether to take or discard an item, but there are
more than these two possibilities:  The algorithm can discard an item, take it
once, take it twice, etc.  No amount of advice can tell it
exactly what to do in the first step because there is an unbounded
number of possible actions.  Does that mean that the online unbounded \kp
cannot be solved exactly by any amount of advice, even when we allow
the number of advice bits to depend on the number of items?  Yes, this
is indeed the case and completely different from the classic 0-1 setting.

\medskip

\begin{theorem}
\label{thm:simpleno1}
Let $f\colon\N\to\N$ be an arbitrary computable function.  The online simple unbounded
\kp cannot be solved exactly with $f(n)$ advice
bits on instances of length~$n$.
\end{theorem}

\begin{proof}
Let $m\in\N$.  We look at the instances $I_k=(1/m-1/m^3, 1-k/m+k/m^3)$, for
$k=0,\ldots,m$.

An algorithm that sees the first item of size~$1/m-1/m^3$ has to decide how
often to take it, which corresponds to a number between $0$ and~$m$.  The only
possibility to be optimal on the instance~$I_k$ is to take the first
item $k$ times, because then the knapsack can be filled completely.  We
have $k\cdot(1/m-1/m^3) + (1-k/m+k/m^3)=1$ and there is no other
possibility to achieve that task.  An algorithm needs sufficient advice to
choose the right one out of $m+1$ possibilities, which requires at
least $\log(m+1)$ advice bits.  The length of the instances $I_k$ is
only $2$, so, if $\log(m+1)>f(2)$, the algorithm will be suboptimal on
at least one~$I_k$.
\end{proof}

\section{The Online General Unbounded \kp}
\label{sec:general}

Finally, let us take a look at the online general unbounded \kp.  Here each
item $x_i$ has a size $0\leq s_i\leq 1$ and a value~$v_i\geq 0$.
Of course, all lower bounds of the online simple unbounded \kp can be transferred immediately
to the more general problem. Furthermore, it turns out that the online general unbounded \kp is
non-competitive for deterministic or randomized algorithms and even for
algorithms that use fewer than $\log n$ advice bits.

\medskip

\begin{theorem}
\label{thm:generalrandom}
The competitive ratio for any randomized online algorithm for the general unbounded \kp is unbounded.
\end{theorem}

\medskip

\begin{theorem}
\label{thm:generaladviceunb}
The competitive ratio of any online algorithm for the general unbounded \kp that uses less than $\log(n)$ advice bits is unbounded.
\end{theorem}

\medskip

Both \cref{thm:generalrandom,thm:generaladviceunb} follow immediately from 
results by B\"ockenhauer et al.\ 
for the 0-1 \kp~\cite{BockenhauerKKR2014}. The examples given in the proofs
used only items of size $1$, so they still hold for the unbounded variant.

It might be of interest to compare the general unbounded \kp and
the general 0-1 \kp with respect to their advice complexity. The latter problem is
very sensitive to how the input is encoded.  One possibility is to
assume that both the sizes and values of items are given as real
numbers and the algorithm is able to do basic operations like
comparisons and arithmetic on these numbers.  To achieve arbitrarily good
competitive ratios with logarithmic advice, a more restrictive input
encoding was used~\cite{BockenhauerKKR2014}. The reason is that
all optimal---and all near-optimal---solutions use a very large number
of items.  Their indices cannot be communicated with a small amount of
advice, and the alternative is to encode sizes of items.

It turns out that the general unbounded \kp is easier to
handle.  The basic property that we will use in the proof of \cref{thm:genup} is
that a near-optimal solution can be built with a constant number of
items.  Their indices can be communicated with logarithmic advice.
Moreover, these items are each used only a constant number of times
with the exception of only one item.  The following intuition shows
why this is the case:  If two items are both packed a large number of
times, both of them have to be tiny.  Then you could just use the
denser one of them.  The resulting new solution cannot be much worse.

\medskip

\begin{theorem}\label{thm:genup}
Let $\epsilon>0$ be a number that does not have to be constant with respect
to the number $n$ of items.
Then the online general unbounded \kp can be solved with a competitive
ratio of at most $1+\epsilon$ with $O((1/\epsilon)\log(n/\epsilon))$ many advice bits.  
\end{theorem}

\begin{proof}
Let $\delta=\epsilon/(\epsilon+2)$.
We fix some optimal solution $S$
to a given instance $I$ of length~$n$.  We say that an item in $S$ is
\emph{small} if its size is at most~$\delta$.  Let $h$ be the total
size of all small items in~$S$ and, if $h>0$, let $x_m\in S$ be a small
item in $S$ with maximum density, i.e., maximizing $v_m/s_m$. Using
$O(\log(1/\delta))$
many advice bits, the oracle tells the algorithm a number $h'$ such that
$h-\delta<h'\leq h$ and, using another $O(\log n)$ advice bits,  it
communicates the index~$m$.

When the algorithm receives $x_m$, it packs it into the knapsack
$\lfloor h'/s_m\rfloor$
times, filling the knapsack at most to a size of $h'$ and packing a
value of at least $v=(h'-\delta)v_m/s_m\geq(h-2\delta)v_m/s_m$.  Let us
compare this to the total value $v'$ in $S$ contributed by small items.
Since $x_m$ is a small item with highest density, $v'\leq
hv_m/s_m$ and $v'-v\leq2\delta v_m/s_m$.

Next, we turn to the items in $S$ that are not small.  There can be at most
$1/\delta$ many of them and each can be packed at most $1/\delta$
times.  We can therefore encode both their indices and their packing
multiplicities by $1/\delta$ numbers each, of sizes at most $n$ and $1/\delta$, respectively,
which can be done with $O((1/\delta)(\log n + \log(1/\delta)))=O((1/\delta)\log(n/\delta))$ many bits.  Receiving this
information as advice allows the algorithm to pack large items exactly
the same as in the optimal solution~$S$.

The algorithm has then packed the same large items and has lost only a
value of $2\delta v_m/s_m$ by suboptimally packing small items.

We have to bound~$v_m/s_m$.  Intuitively, $v_m/s_m$ cannot be much
smaller than the average density of all items in~$S$.  Otherwise, we
could pack $x_m$ as often as possible into an empty knapsack achieving
a gain of at least $(1-\delta)v_m/s_m$, which cannot be larger
than~$\opt(I)$.

Hence, we can assume from now on that $(1-\delta)v_m/s_m\leq\opt(I)$,
which implies
\[
2\delta v_m/s_m\leq 2\delta\opt(I)/(1-\delta)=\epsilon\cdot\opt(I)\;,
\]
which is an upper bound on the gap between the gain of our algorithm
and the optimal solution.
\end{proof}

\section{Conclusion}
\label{sec:conclusion}

We have analyzed the hardness of the unbounded \kp for online algorithms
with and without advice, both in the simple and general case, and found
some significant differences to the 0-1 \kp. A simple greedy strategy
achieves a constant competitive ratio for the simple unbounded \kp even in
the deterministic case. Moreover, unlike for the simple 0-1 \kp, an
unlimited amount of randomization helps improving the competitive ratio for
the online simple unbounded \kp. 

It remains as an open problem to find out the exact competitive
ratio in expectation that can be achieved for the online simple unbounded \kp by randomized
algorithms. It would also be interesting to see, both for the simple and
the general case, how exactly the transition from sublogarithmic advice to
logarithmic advice looks like: What competitive ratio can be achieved with
$\alpha\log n$ advice bits, for some fixed constant $\alpha$?

For the online 0-1 \kp, considering variants where the algorithm is allowed to
remove packed items or reserve items for packing them later on has given
valuable structural insights. It would for sure be interesting to look at
such variants also for the online unbounded \kp.

\backmatter
\bmhead{Acknowledgments}
We would like to thank an anonymous reviewer for valuable suggestions regarding \Cref{thm:adviceuptolog}. 

Ralf Klasing was partially supported by the ANR project TEMPOGRAL
(ANR-22-CE48-0001).  Matthias Gehnen, Henri Lotze, and Daniel Mock were
partially supported by an IDEA League Short-Term Research Exchange
Grant.

\begin{appendices}

\newpage
\section{Technical Result for \cref{thm:randup}}

In the proof of \cref{thm:randup}, we need the following result that can be
proven by methods from standard calculus. We assume the preconditions and
notation from \cref{thm:randup}.

\begin{lemma}
\label{lemma:increase}
The function 
\begin{equation*}g(\xi)=\frac{\xi}{\frac{2}{3}\phalf+\frac{2}{3}p_{2/3} +\int_{\frac{2}{3}}^{\xi}x\cdot f_X(x)\d x}
\end{equation*}
is increasing for $2/3\leq \xi\leq 1$.
\end{lemma}
\begin{proof}
We will show that 
\begin{align*}
g'(\xi)=\frac{\left(\frac{2}{3}\phalf+\frac{2}{3}p_{2/3} +\int_{\frac{2}{3}}^{\xi}x\cdot f_X(x)\d x\right)-\xi^2f_X(\xi)}{\left(\frac{1}{2}\phalf+\int_{\frac{1}{2}}^{\frac{2}{3}} x\cdot f_X(x)\d x+\frac{2}{3}p_{2/3} +\int_{\frac{2}{3}}^{\xi}x\cdot f_X(x)\d x\right)^2}\geq 0
\end{align*}
for $\frac{2}{3}\leq \xi\leq 1$. This is equivalent to showing that
\begin{equation*}
h(\xi)=\frac{2}{3}\phalf+\frac{2}{3}p_{2/3} +\int_{\frac{2}{3}}^{\xi}x\cdot f_X(x)\d x-\xi^2f_X(\xi)\geq 0
\end{equation*}
for $\frac{2}{3}\leq \xi\leq 1$. When $\xi=\frac{2}{3}$, we check that 
\begin{equation*}
\frac{2}{3}\phalf+\frac{2}{3}p_{2/3}
-\frac{4}{9}f_X\biggl(\frac{2}{3}\biggr)\geq
\frac{2}{3}\phalf-\frac{4}{9}f_X\biggl(\frac{2}{3}\biggr)\geq 0,
\end{equation*}
since 
\[
\frac{2}{3}\phalf-\frac{4}{9}f_X\biggl(\frac{2}{3}\biggr)\geq 0\iff
1-\frac{2}{3}\cdot\frac{1-\ln(\frac{4}{3})}{\frac{4}{3}}\geq 0\iff
1+\ln\biggl(\frac{4}{3}\biggr)\geq 0\;.
\]
Finally, we check that $h(\xi)$ is itself increasing for $\frac{2}{3}\leq \xi\leq 1$, since
\begin{align*}
h'(\xi)&=\xi\cdot f_X(\xi)-\left(2\xi \cdot f_X(\xi)+\xi^2\cdot f'_X(\xi)\right)\\
&=-\xi\cdot f_X(\xi)-\xi^2\cdot\frac{-\frac{2\xi}{2-\xi}-2(1+\ln(2-\xi))}{4\xi^2}\\
&=-\frac{1+\ln(2-\xi)}{2}+\frac{\xi}{2-\xi}-\frac{1+\ln(2-\xi)}{2}\\
&=\frac{\xi}{2-\xi}\\
&\geq 0
\end{align*}
for $2/3\leq \xi\leq 1$.
\end{proof}
\end{appendices}

\end{document}